\documentclass[10pt,conference,letterpaper]{IEEEtran}
\usepackage{algpseudocode}
\usepackage{algorithm}
\usepackage{amsmath,amsthm,amsfonts}
\usepackage{graphicx}
\usepackage{subfigure}
\usepackage{epstopdf}
\usepackage{url}
\usepackage{color}
\usepackage{enumerate}

\newtheorem{lemma}{Lemma}
\newtheorem{theorem}{Theorem}
\newtheorem{corollary}{Corollary}

\DeclareMathOperator*{\argmax}{arg\,max}
\title{Topology Discovery Using Path Interference }
\author{Anurag Rai, Eytan Modiano\\ Laboratory for Information and Decision Systems, MIT, USA \\ \{rai,modiano\}@mit.edu \vspace{-2ex}
}

\begin{document}
\maketitle

\begin{abstract}
We consider the problem of inferring the topology of a network using the measurements available at the end nodes, without cooperation from the internal nodes. To this end, we provide a simple method to obtain path interference which identifies whether two paths in the network intersect with each other. Using this information, we formulate the topology inference problem as an integer program and develop polynomial time algorithms to solve it optimally for networks with tree and ring topologies. Finally, we use the insight developed from these algorithms to develop a heuristic for identifying general topologies. Simulation results show that our heuristic outperforms a recently proposed algorithm that uses distance measurements for topology discovery.
\end{abstract}

\section{Introduction}

Knowing the topology of the underlying network can provide several advantages to the communicating hosts. For example, the topology can be used to improve the throughput and robustness of the network \cite{ron, jones}, and it can be a necessary part of identifying bottlenecks and critical links in the network \cite{survivability}. It can also be used to monitor the network or to simply get a picture of the underlying system. However, often the owners of the network keep the topology information hidden due to privacy and security concerns \cite{sarac}. This has led to a significant amount of research on topology discovery. We develop a new method that can be used to identify general network topologies. This method only requires the interference pattern of the paths in the network which can be inferred from the data available at the end nodes.

Prior work on topology discovery can be divided into two main categories: algorithms that require cooperation from the internal nodes and the algorithms that do not. Many algorithms for topology discovery, usually designed for the purpose of mapping the Internet, use ICMP commands like traceroute \cite{rocketfule, crovella, sarac}. These methods requires some level of cooperation from the network providers. The other methods, that fall under the category of network tomography \cite{vardi,survey}, use data that can be measured directly at the end nodes. Our method falls under this category as we do not seek any information from the internal nodes.

In the network tomography literature significant attention has been given to the discovery of tree networks. Papers such as \cite{nowak2, towsley, yang} use probing mechanisms to infer single source multiple destination trees. There is also some work on combining these single source trees to form a multi-source multi-destination network \cite{nowak3, sattari}. In \cite{singh}, the authors provide a method for identifying minimal trees with multiple sources and multiple destinations by using distance measurements. 

In \cite{kelner}, the authors develop an algorithm called RGD1 that attempts to  discover a general network topology. It uses a set for four nodes that share a link, called quatrets and uses them to build an approximation of the entrie network.
The discovery of the quatret  and placement of the nodes in the topology requires the shortest path distance between the nodes, which is inferred using packet delay. RGD1 algorithm is very close to our algorithm in terms the objective, hence we will compare its performance against ours via simulation.

In order to obtain the interference pattern, we provide a simple method based on linear regression. This method uses the number of in-flight packets in the paths and the delay experienced by the packets to determine whether a given pair of paths interfere with each other. Using the resulting interference information, we formulate the topology inference problem as an integer program. We develop polynomial time algorithms to solve it optimally for networks with special topologies, namely tree or ring topology. Both of these algorithms obtain the minimal version of the network, even when the original network is not minimal.
We also develop a heuristic that attempts to recover any general topology in polynomial time. 

The main contributions of this paper can be summarized as follows:
\begin{itemize}
\item We use the interference pattern of the paths to formulate an integer linear program (ILP) that obtains the network that has the fewest number of links and supports the given interferences. The solution provides a new method to discover a general network topology. 
\item We provide an upper bound, a lower bound and a sufficient condition for optimality for the ILP.
\item We design two polynomial time algorithms to recover tree and ring networks and show that if the network is in fact a tree or a ring, the algorithms solve the ILP optimally. 
\item Building upon the tree and the ring algorithms, we develop a polynomial time heuristic to identify general networks. Using simulations we show that this method outperforms the RGD1 algorithm of \cite{kelner}.
\end{itemize}

\section{Model}
\subsection{Network Model}
We model the network as a graph $G=(N,E),$ where $N$ is the set of nodes and $E$ is the set of edges. We assume that all the links in the network are bidirectional and have unit capacity. Each bidirectional link $\{i,j\}$ is composed of two directed links $(i,j)$ and $(j,i)$. The network has two types of nodes: the overlay nodes, which represent hosts and can be controlled, and the underlay nodes, which represent routers that are uncontrollable and do not provide any direct feedback. We represent the set of overlay nodes by $\mathcal{O}$ and the set of underlay nodes by $U$, and $N = \mathcal{O} \cup U$. We further assume that each overlay node is connected to only one underlay node. Other that this, we do not have any knowledge of the structure of the network. The main goal of this paper is to recover the graph $G$ from data measured at the overlay nodes.

All the overlay nodes are connected to each other by tunnels, which are paths that go through the underlay nodes. A tunnel $l = (l_1, l_2, ..., l_{|l|})$ consists of overlay nodes $l_1$ and $l_{|l|}$ and the rest of the nodes are underlay. Since, $l_1$ and $l_{|l|}$ are connected to only one underlay nodes, we will often refer to node $l_2$ as the parent of node $l_1$, $p(l_1)$, and node $l_{|l|-1}$ as the parent of node $l_{|l|}$, $p(l_{|l|})$. There are a total of $L = |\mathcal{O}| \times (|\mathcal{O}|-1)$ tunnels in the network. 


We also assume that each node $i \in N$ maintains a queue for each of it outgoing link $(i,j) \in E$. Packets from all the tunnels that uses the link $(i,j)$ gets enqueued in this queue when they reach node $i$ and are served on a first come first serve basis. 

\subsection{The Interference Matrix $\mathcal{F}$}
Our algorithm for recovering the graph $G$ is based on whether or not any two tunnels between the overlay nodes intersect with each other. In order to identify this we propose a simple method based on linear regression. We note that depending on the measurements available, other methods such as the ones from \cite{nowak,kelner} can also be used to derive this information.

Let $d_l(t)$ represent the delay experienced by a packet that enters tunnel $l$ at time $t$. Tunnels in the network can intersect with each other, hence, the path traversed by a tunnel can have packets belonging to itself and packets from other tunnels. Let $h_l(t)$ represent the number of packets that belong to tunnel $l$ that are still in the tunnel at time $t$. We will refer to these packets as packets in flight of tunnel $l$. The delay experienced by a packet entering tunnel $l$ at time $t$ is affected by the number of packets in that tunnel and other tunnels that intersect with it. Considering only a pair of tunnels $k$ and $l$, we can model the relationship between the packets in flight and delay as a linear function:
$$d_l(t) = h_l(t) + \alpha_{kl} h_k(t) + \eta_l.$$ 
Here $\alpha_{kl}$ represents the fraction of packets of tunnel $k$ that are in the path traversed by tunnel $l$ and $\eta_l$ is random perturbation (noise). 

By injecting randomly generated traffic into each pair of tunnels and measuring the packet delay and packets in flight, it is possible to determine if two tunnels intersect. In particular, using linear regression it is possible to calculate the optimal parameters $\alpha_{kl}$ that minimizes the noise for each pair of tunnel $(k,l)$.
When tunnels $l$ and $k$ do not intersect, the number of packets in tunnel $k$ does not affect the delay of packets entering tunnel $l$, hence, $\alpha_{kl} \approx 0.$ Otherwise, $\alpha_{kl}$ is closer to 1. We use these $\alpha_{kl}$ values to create the $L \times L$  binary interference matrix $\mathcal{F}$. If $\alpha_{kl} \approx 0$ then $\mathcal{F}_{kl} = 0$, and $\mathcal{F}_{kl} = 1$ otherwise. Moreover, $\mathcal{F}$ is symmetric, implying $\mathcal{F}_{kl} = \mathcal{F}_{lk}$.

We will use the graph representation of $\mathcal{F}$ in some of our results. We refer to such a graph as the interference graph of the network, $G_{\mathcal{F}}(N_{\mathcal{F}}, E_{\mathcal{F}})$. This graph is simply the graph formed by using $\mathcal{F}$ as an adjacency matrix, where $N_{\mathcal{F}}$ consists of tunnels and an edge exists between tunnels that interfere with each other.  An example of an interference matrix and its corresponding graph is given in Figure \ref{fig: minimality}.


\subsection{Minimal topology} \label{sec: minimality}
There exist many networks that produce the same interference matrix $\mathcal{F}$, hence, these networks are indistinguishable by our method. For example, each tunnel in the two networks shown in Figure \ref{fig: minimality} face the same interference. E.g. the tunnel $(1,...,2)$ only interferes with tunnel $(1,...,3)$ in both the networks. Hence, they produce the same $\mathcal{F}$ matrix. We are interested the smallest network, in terms of the number of links, that produces the given $\mathcal{F}$ matrix. We will call such a topology the minimal network topology. 

\begin{figure}[h]
    \centering
    \subfigure[Not minimal]{
        \centering
        \includegraphics[scale=.6]{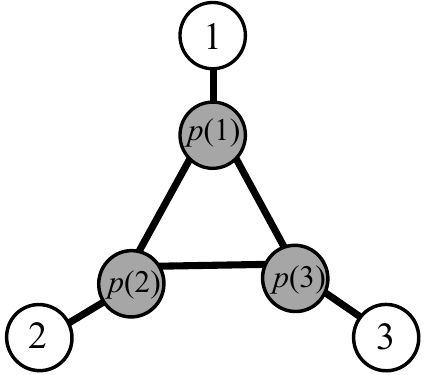}
        \label{fig: not_minimal}
	}
    \subfigure[Minimal]{
        \centering
        \includegraphics[scale=.6]{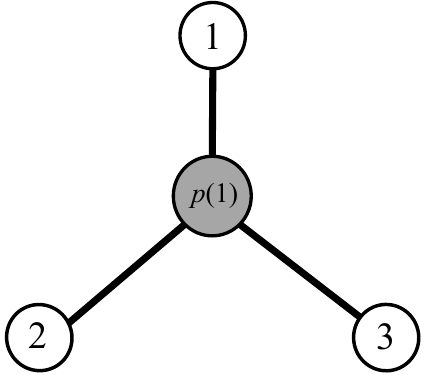}
        \label{fig: minimal}
	}
    \subfigure[The interference matrix $\mathcal{F}$.]{
        \centering
        \includegraphics[scale=.7]{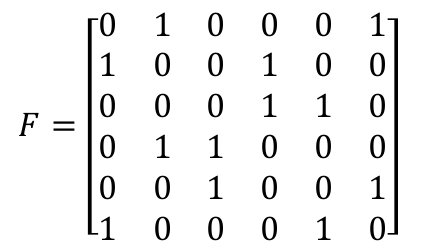}
        \label{fig: F}
	}
    \subfigure[The interference graph $G_\mathcal{F}$. Each node $(i,j)$ represents a tunnel $(i,...,j)$.]{
        \centering
        \includegraphics[scale=.7]{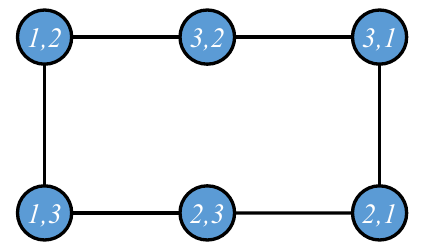}
        \label{fig: F}
	}
    \caption{Two topologies in Figures \ref{fig: not_minimal} and \ref{fig: minimal} produce the same $\mathcal{F}$ matrix. The white nodes are overlay and gray nodes are underlay, and the network uses the shortest path routing.}
	    \label{fig: minimality}
\end{figure}

A necessary condition for a network to be minimal was identified in \cite{kelner}. Specifically, all underlay nodes must have at least three neighbors. If an underlay has only one neighbor, we can simply remove it to obtain a smaller network that is indistinguishable from the original network by using only the measurements available at the overlay. If an underlay node has two neighbors, we can connect its two neighbors and remove the node in order to obtain a smaller network with the same properties. We note that this condition is not sufficient for minimality in general. E.g. in Figure \ref{fig: not_minimal}, all the underlay nodes have 3 neighbors but the topology is not minimal. We will provide a sufficient condition for minimality, and show that the necessary condition from \cite{kelner} is also sufficient for specific topologies, namely trees and rings.

In this paper we assume that the $\mathcal{F}$ matrix for a network $G(N,E)$ is given (i.e. obtained via measurements, as described earlier) and focus on obtaining the minimal network $\hat{G}^*(\hat{N}^*, \hat{E}^*)$ that supports this interference pattern.

\section{Integer Programming Formulation}
We formulate the problem of finding the minimal network that supports the given path interference pattern as an integer linear problem (ILP). Although a solution for this ILP is computationally intractable for large networks, studying this formulation will provide us with useful insights into the problem. Also, when the network is small, we are able to solve it optimally.

\subsection{Integer program}
Let us consider a network with $|\hat{N}|$ nodes.  Nodes $1,..., |\mathcal{O}|$ are overlay nodes and nodes $|\mathcal{O}| + 1, ..., |\hat{N}|$ are underlay nodes. Note that the set $\mathcal{O}$ is known a priori. 

Let $x_{ij}^l \in \{0,1\}$ represent whether link $(i,j)$ is used by tunnel $l$, for $1\le i,j \le |\hat{N}|$, $1 \le l \le L$, and $x^l_{ii} = 0 \, \forall l$. For notational simplicity, we define another variable $x_{ij}$ which represents whether the link $\{i,j\}$ is used by any tunnel in either direction. Hence, 
\begin{equation}
x_{ij} = \lor_{l}x_{ij}^l \lor x_{ji}^l, \qquad\qquad \forall i,j
\end{equation}

Here ``$\lor$'' is a logical OR operator. Note that such logical constraints can easily be transformed into a set of linear (integer) constraints \cite{optimization_book}. The objective function can be written as 
$$\text{minimize} \sum_{ij} x_{ij}.$$

Our network model assumes that each overlay node is connected to only one underlay link. This can be enforced by using the following constraint:
\begin{equation}
\sum_j x_{ij} = 1, \qquad  i =  1, ..., |\mathcal{O}|.
\end{equation}

Again, to simplify the notation we define two new variables. Let $s(l,j) \in \{0,1\}$ represent whether tunnel $l$ begins at node $j$, and let $d(l,j)$ represent whether tunnel $l$ ends at node $j$. These values are known a priori, so we can replace these variables with their respective values while formulating a specific problem.
Now we can write the next set of constraints which are essentially the flow conservation constraints. These constraints guarantee that each tunnel has a set of connected links in the network, starting and ending at its respective overlay nodes.
\begin{align}
\sum_i x^l_{ij} + s(l,j) = \sum_i x^l_{ji} + d(l,j), \qquad &j = 1, ..., |\hat{N}|, \nonumber\\
& l = 1, ..., L
\end{align}

We can see that the flow conservation constraints above allows loops to be formed in the network. Unlike max-flow type problems where loops can be removed in the post processing without harming the feasibility, removing them in our case can change the interference pattern of the tunnels. Hence we need to add constraints to avoid formation of loops.

Similar problems arise in the ILP formulation of the Travelling Salesman Problem (TSP). We use the technique originally proposed by Miller-Tucker-Zemlink in \cite{TSP} to resolve this issue in TSP and add the following constraints:
\begin{align}
u^l_i - u^l_j + |\hat{N}| x^l_{ij} \le |\hat{N}| -1, \qquad \forall i\ne j, l = 1, ..., L.
\end{align}
Here, the variables $u^l_i \ge 0$ is used to assign an order to each node $i$ in each tunnel $l$. If $x^l_{ij} =1,$ then $u^l_j \ge u^l_i + 1,$ so the next node $j$ is assigned a higher value than node $i$. 
Otherwise, $u_i^l - u^l_j \le |\hat{N}| - 1.$ This ensures that there are enough values to assign to all the nodes that the tunnel might pass through.

Finally we consider the interference constraints. For each tunnel pair $(k,l)$ we add a set of constraints depending on whether tunnels $k$ and $l$ interfere with each other. If tunnels $k$ and $l$ do not interfere we have the following constraints:
\begin{align}
x^k_{ij} + x^l_{ij} \le 1, \qquad \forall i,j, \text{ and } k,l: F(k,l) = 0.
\end{align}
This ensures that two tunnels that do not interfere with each are never assigned to the same link.
If $F(k,l) = 1$, then both the tunnels $k$ and $l$ must appear together in at least one of the links. We enforce this with the following constraints
\begin{align}
\sum_{i,j} x^k_{ij} \land x^l_{ij} \ge 1, \qquad \forall i,j, \text{ and } k,l: F(k,l) = 1.
\end{align}
Here ``$\land$'' is the logical $AND$ operator, and these constraints can also be transformed into a set of linear (integer) constraints.

The objective function along with the constraints (1) through (6) give the required ILP for identifying a minimal network. After solving the ILP, the graph can be recovered from the links $\{i,j\}$ for which $x_{ij} = 1$. A node that is not used by any of the tunnels can simply be removed from the recovered network.

\subsection{Example}
We consider a network where 6 underlay nodes are arranged to form a $3\times 2$ grid, and an overlay node is attached to each underlay node. The network uses the shortest path routing. The $30 \times 30$ interference matrix $\mathcal{F}$ is generated by determining whether two paths intersect with each other. We formulate the ILP with $|\hat{N}| = 12$ then solve it using the Gurobi solver \cite{gurobi}. 

\begin{figure}[h]
    \centering
    \subfigure[Original graph $G$.]{
        \centering
        \includegraphics[scale=.7]{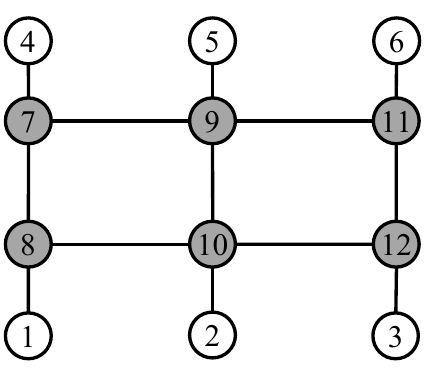}
        \label{fig: original1}
	}
    \subfigure[Recovered graph $\hat{G}$.]{
        \centering
        \includegraphics[scale=.7]{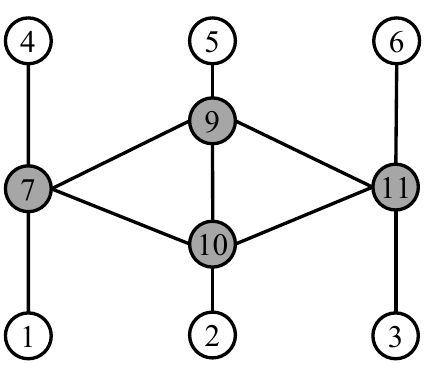}
        \label{fig: recovered1}
	}
    \caption{Recovering a network topology by solving the ILP.}
	    \label{fig: example1}
\end{figure}

The original and the recovered network are shown in Figure~\ref{fig: example1}. The recovered network has fewer nodes and edges than the original network. Link $\{7,8\}$ in the original network is used only by tunnels $(1,8,7,4)$  and $(4,7,8,1)$ in different directions. Hence there is no interference on this link, and it can be removed without changing the interference matrix. For the same reason, link $\{11,12\}$ can be removed to obtain the minimal network.
Even after the removal of the links, we can see that the recovered network looks quiet similar to the original.

\subsection{Upper bound}
We provide an upper bound on the solution to the ILP in the previous section by using a simple algorithm given in Algorithm \ref{alg: feasible}. This algorithm produces a feasible solution to the ILP by assigning two interfering tunnels to a new link in $\hat{G}$. This algorithm can be suboptimal because in an optimal solution many tunnels can interfere at the same link.

\begin{algorithm}[h!]
\caption{FeasibleGraph$(\mathcal{F}, \mathcal{O})$ for obtaining a feasible network $\hat{G}$}
\label{alg: feasible}
\flushleft{
\begin{enumerate}
\item Create a graph $\hat{G}$ with $|E_\mathcal{F}|$ edges in a line.
\item For each link $\{k,l\} \in E_\mathcal{F}$ assign tunnel $k$ and $l$ to a unique edge in $\hat{G}$. All the tunnels traverse the line graph in the same direction.
\item Connect links in $\hat{G}$ that have the same tunnel, if they aren't already connected, such that the tunnels form a loop free path. This can require either a single link, or a node and two links; see example in Figure \ref{fig: eg_greedy}.
\item Add nodes $\hat{\mathcal{O}} = \{\hat{o}_1, \hat{o}_2, ..., \hat{o}_{|\mathcal{O}|}\}$ to $\hat{G}$. Each node $\hat{o}_i \in \hat{\mathcal{O}}$ corresponds to an overlay node $o_i \in \mathcal{O}.$
\item For each node $\hat{o}_i$ add a parent node $p(\hat{o}_i)$ and edge $\{\hat{o}_i, p(\hat{o}_i) \}$ to $\hat{G}$.
\item For each tunnel $l$ that starts at $o_i$ assign tunnel $l$ to link $(\hat{o}_i, p(\hat{o}_i))$. 
\item For each tunnel $l$ that ends at  $o_i$ assign tunnel $l$ to link $(p(\hat{o}_i), \hat{o}_i)$.
\item Complete the tunnels by connecting $p(\hat{o}_i)$ to the partial tunnels formed in Step 3.
\end{enumerate}
}
\end{algorithm}

Algorithm \ref{alg: feasible} starts with a $\hat{G}$ that is a line graph  with $|E_\mathcal{F}|$ edges, then maps each link in the interference graph $G_{\mathcal{F}}$  to a link in $\hat{G}$. Each edge in $G_{\mathcal{F}}$ represents two tunnels that pass through the same edge in $G$, so if there is an edge between tunnels $k$ and $l$ in $G_{\mathcal{F}}$, then tunnels $k$ and $l$ are assigned to one of the links in $\hat{G}$. When all the interferences are assigned, it is likely that the same tunnel gets assigned to links that are not attached to each other. In such a case, new links are added to create complete tunnels. An example of this process (Steps 1-3) is given in Figure \ref{fig: eg_greedy}. At the end of Step 3, all the interference constraints are satisfied. Steps 4-8 add the overlay nodes and makes sure that each overlay node is connected to a single underlay node.

\begin{figure}[h]
    \centering
    \subfigure[Interference graph for tunnels $a$, $b$, $c$ and $d$.]{
        \centering
        \includegraphics[scale=.6]{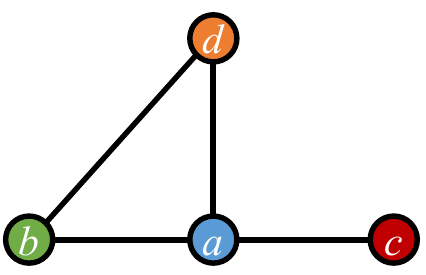}
        \label{fig: interference_greedy}
	}
    \subfigure[$\hat{G}$ after Step 2. Every pair of interfering tunnels are assigned to some link in the graph.]{
        \centering
        \includegraphics[scale=.8]{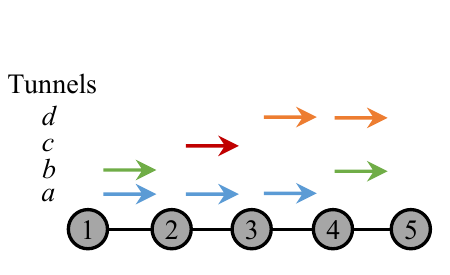}
        \label{fig: step2_greedy}
	}
 \subfigure[$\hat{G}$ after Step 3. Tunnels that are disconnected are connected using extra nodes and edges. ]{
        \centering
        \includegraphics[scale=.8]{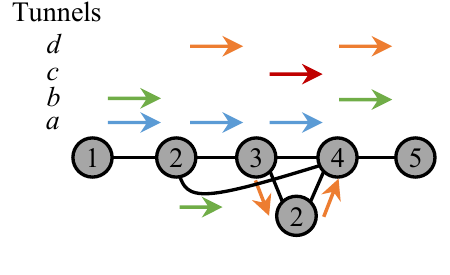}
        \label{fig: step3_greedy}
	}
    \caption{Example of an execution of Steps 1 to 3 of Algorithm 1.}
	    \label{fig: eg_greedy}
\end{figure}

We give the following lemma to show that Algorithm \ref{alg: feasible} produces a feasible solution to the ILP. Then Theorem \ref{thm: upper bound} establishes the upper bound on the number of links used by this algorithm.
\begin{lemma}
Algorithm \ref{alg: feasible} obtains a feasible solution to the ILP in Section 1.
\end{lemma}
\begin{proof}
The proof of this lemma is given in Appendix \ref{app: lemma 1}.
\end{proof}

\begin{theorem} \label{thm: upper bound}
The number of edges required for a feasible solution of the ILP, $|\hat{E}| \le |E_\mathcal{F}| + 2L|E_{\mathcal{F}}| + |O| +  2L$.
\end{theorem}
\begin{proof}
The proof of this theorem is given in Appendix \ref{app: theorem 1}.
\end{proof}

\subsection{Lower bound}
We establish a lower bound on the number of edges in the minimal graph using the properties of the interference graph. In order to minimize the number of links, we want to assign as many interfering tunnels as possible to the same link. However, we cannot have two tunnels be assigned to the same link if they don't interfere with each other. This property is nicely abstracted by the cliques in the interference graph $G_\mathcal{F}$. The tunnels, represented by the nodes in $G_\mathcal{F}$, that are in the same clique interfere with each other. So we can assign all of them to the same link. A lower bound is given by the minimum number of cliques required to cover all the links. For example, two cliques are needed to cover all the edges of the interference graph in Figure \ref{fig: interference_greedy}, so we need at least two links in $\hat{G}$ to represent all the interferences. In graph theory the smallest such set is known as the {\em minimum edge clique cover}\footnote{This is different from the minimum node clique cover which is the smallest set of cliques required to cover all the nodes.}, and the size of such set is known as the {\em intersection number} of the graph \cite{roberts}. Computing the minimum edge clique cover of a graph is known to be NP hard so it might not be useful for the purpose of comparing our solutions. However, in the next subsection we will use it to derive conditions when a recovered graph achieves the lower bound and guarantee optimality.

The following lemma presents the lower bound result in terms of the number of directed links required to have a feasible solution. Theorem \ref{thm: cliques bound} extends this result to the case with undirected links, which  is the setup in this paper.

\begin{lemma} \label{lem: lower bound}
Let $|\hat{E}_D|$ be the number of directed links required for a feasible solution of the ILP.
Let C be the size of the minimum edge clique cover for the interference graph $G_\mathcal{F}$.
Then $|\hat{E}_D | \ge C$.
\end{lemma}
\begin{proof}
The proof of this lemma is given in Appendix \ref{app: lemma 2}.
\end{proof}

\begin{theorem} \label{thm: cliques bound}
Let $|\hat{E}|$ be the number of undirected links required for a feasible solution of the ILP.
Let C be the size of the minimum edge clique cover for the interference graph $G_\mathcal{F}$.
Then, $$|\hat{E}| \ge {C \over 2}.$$
\end{theorem}
\begin{proof}
The proof of this theorem is given in Appendix \ref{app: theorem 2}.
\end{proof}

\subsection{A sufficient condition for optimality}
We give a condition under which a the recovered network has the same number edges as the original network. When this condition is satisfied, the interference pattern cannot be achieved in a smaller network, so this result also provides a sufficient condition for minimality of a network. We prove this result by showing that if the condition is satisfied, then the recovered network achieves the lower bound developed in the previous subsection. We use this result in the subsequent sections to show that our polynomial time algorithms optimally solve the ILP for special networks.

The main result of this subsection says that a given network is minimal if every directed edge in the network is associated with a unique interference (interfering pair of tunnels). Intuitively, this condition seems reasonable because if it is satisfied then each directed link in the graph creates a unique clique in the minimum edge clique cover of the interference graph. 



\begin{lemma} \label{lem: cliques equality}
 The size of the minimum edge clique cover of $G_\mathcal{F}$, $C=2|E|$ if and only if for each directed edge $(i,j)$ there exists a pair of tunnels $k^{ij}$ and $l^{ij}$ such that they intersect at link $(i,j)$ and nowhere else.
\end{lemma}
\begin{proof}
The proof of this lemma is given in Appendix \ref{app: lemma 3}.
\end{proof}

\begin{theorem} \label{thm: minimality}
Let C be the size of the minimum edge clique cover for the interference graph $G_\mathcal{F}$. Let $\hat{G}^*(\hat{N}^*, \hat{E}^*)$ be the optimal network obtained by solving the ILP.  If every directed link $(i,j)\in E$ has a pair of tunnels $k^{ij}$ and $l^{ij}$ such that they intersect at link $(i,j)$ and nowhere else, then $\hat{G}^*$ has the same number of edges as the original network, i.e. $|\hat{E}^*| = |E|$.
\end{theorem}
\begin{proof}
The proof of this theorem is given in Appendix \ref{app: theorem 3}.
\end{proof}

Note that this theorem provides a sufficient condition but it may not be necessary. That is, there may be graphs where $C < 2|E|$ but the ILP still produces a graph with $|E|$ edges. Also, if the number of edges in the optimal network obtained by solving the ILP is the same as the original network, then we know that the both the networks are minimal. Hence, we can use the condition in the theorem as a sufficient condition for minimality of a network.

\begin{corollary}
A network $G(N,E)$ is minimal if every directed link $(i,j)\in E$ has a pair of tunnels $k^{ij}$ and $l^{ij}$ such that they intersect at link $(i,j)$ and nowhere else.
\end{corollary}

\section{Identifying Trees} \label{sec: tree}
We design a polynomial time algorithm to recover a tree network. If $G$ is a minimal tree, i.e. every non leaf nodes have at least three neighbors and all the leaf nodes are overlay, then this algorithm recovers the tree exactly. A similar result on recovering trees by using distance between the leaf nodes is given in \cite{singh}, however, the algorithm of \cite{singh} requires the network to be minimal. In the situation when the network $G$ is a non-minimal tree, our algorithm produces a $\hat{G}$ that is a minimal tree corresponding to $G$ since both the networks have the same $\mathcal{F}$ matrix. Note that there is a unique minimal tree corresponding to each non-minimal tree which can be obtained by using the process discussed in Section \ref{sec: minimality}.

\subsection{Algorithm}
%
%
The tree identification algorithm is given in  Algorithm \ref{alg: tree}. The algorithm uses the interference matrix $\mathcal{F}$ to obtain a tree graph $\hat{G}$ with the same $\mathcal{F}$. It begins by initializing the graph $\hat{G}$ and checking for terminating conditions in Steps 1 to 3. In Step 4, the algorithm identifies a node $k_1^*$ such that when all its siblings along with itself are removed, its parent becomes a leaf node. This property will later help us compute a new $\mathcal{F}$ matrix of the reduced graph. In Step 5, this algorithm finds a group of nodes $X_{k^*}$ that consists of all the sibling nodes of $k_1^*$. Procedure 3 is used to identify such nodes; see Lemma \ref{lem: same parent} for proof. These nodes are then added to the recovered graph $\hat{G}$ in Step 6 by assigning then a common parent node, $p(X_k^*)$.

Steps 7 removes the sibling nodes in $X_k^*$ from the original network $G$. Since the graph $G$ is not available, the removal is done indirectly by removing the corresponding tunnels from the $\mathcal{F}$ matrix. Note that node $k_1^*$ is not removed, instead it is renamed as the parent of the group $p(X_k^*)$ in Step 8. This works because when all the siblings of $k_1^*$ are removed, the interference of the tunnels that start or end at $k_1^*$ is the same as the tunnels that start or end at its parent node. The algorithm is iteratively applied to the reduced $\mathcal{F}$ matrix until only one or two leaf nodes remain. 

%
%

\begin{algorithm}[h!]
\caption{IdentifyTree$(\mathcal{F}, \mathcal{O})$ for recovering a tree network}
\label{alg: tree}
\begin{enumerate}
\item Add the nodes in $\mathcal{O}$ to $\hat{G}$.
\item If $|\mathcal{O}| = 1$ return $\hat{G}$ .
\item If  $|\mathcal{O}| = 2$, add an edge between the two nodes in $\hat{G}$ and return $\hat{G}$ .
\item Identify the tunnel $k^*$ that interferes with the largest number of other tunnels, $k^* = \argmax_k \sum_l \mathcal{F}_{kl}$. Let $k^*_1$ be the first node of tunnel $k^*$.
\item For each node $i \in \mathcal{O}$, use Procedure \ref{alg: isSibling} to decide whether it has the same parent as $k^*_1$. Let $X_{k^*}$ be the set of nodes that successfully pass the test.
\item Add a new node $p({X_{k^*}})$ to $\hat{G}$. Connect $p({X_{k^*}})$ to the nodes in $X_{k^*}$ in $\hat{G}$.
\item For each node $i\in X_{k^*}, i\ne k^*_1$:
	\begin{itemize} 
		\item Remove rows and columns corresponding to all the tunnels starting or ending at $i$ from $\mathcal{F}$.
		\item Remove node $i$ from $\mathcal{O}$.
	\end{itemize}
\item Rename node $k^*_1$ to $p({X_{k^*}})$ so that any tunnel in $\mathcal{F}$ starting or ending at $k^*_1$ starts or ends at $p({X_{k^*}})$ respectively.
\item Goto Step 2.
\end{enumerate}
\end{algorithm}

An example of the graphs created after the first and the second iterations of this algorithm are shown in Figure \ref{fig: tree example}. In the first iteration, Step 4 identifies one of the tunnels that intersect with the most number of other tunnels, (5,...,1). So $X_k^* = \{5,6,7\}$ in obtained in Step 5. This avoids obtaining sibling groups such as $\{3,4\}$, which when removed does not make their parent a leaf node. Step 6 produces the $\hat{G}$ shown in Figure \ref{fig: recovered tree}, and Steps 7 and 8 result in the reduced tree shown in Figure \ref{fig: tree 1iteration}. The $\mathcal{F}$ matrix of the reduced tree is obtained by removing all the tunnels with nodes 6 and 7, then renaming node 5 to the parent node $p(5,6,7)$. Similarly the result of the second iteration is shown in Figures \ref{fig: tree 2iteration} and \ref{fig: recovered tree2}. Since there is only one group of siblings left in the graph after this iteration, the third iteration results in the $G$ with only one node. Also, the third iteration produces the $\hat{G}$ that is identical to the original graph in Figure \ref{fig: original tree}.

\begin{figure}[h]
    \centering
    \subfigure[Original network $G$.]{
        \centering
        \includegraphics[scale=.7]{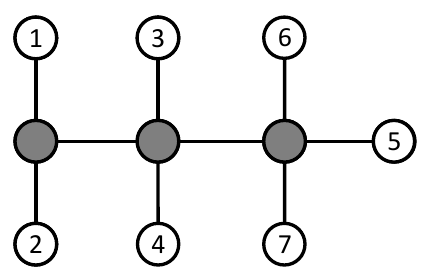}
        \label{fig: original tree}
	} \\
    \subfigure[Graph $G$ (implied by $\mathcal{F}$) after the first iteration.]{
        \centering
        \includegraphics[scale=.7]{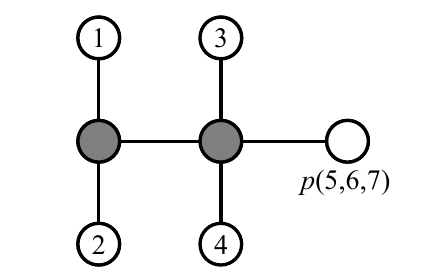}
        \label{fig: tree 1iteration}
	}
 \subfigure[Recovered graph $\hat{G}$ after the first iteration.]{
        \centering
        \includegraphics[scale=.7]{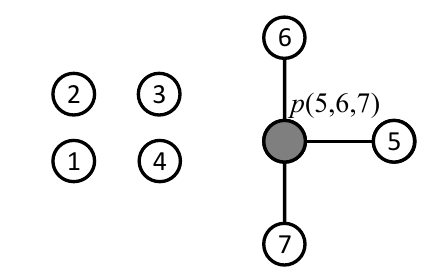}
        \label{fig: recovered tree}
	}
    \subfigure[Graph $G$ (implied by $\mathcal{F}$) after the second iteration.]{
        \centering
        \includegraphics[scale=.7]{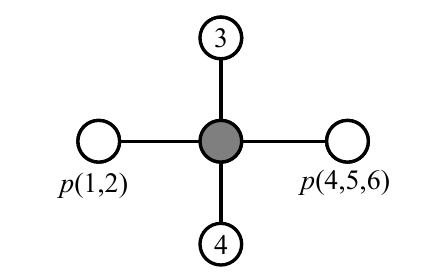}
        \label{fig: tree 2iteration}
	}
 \subfigure[Recovered graph $\hat{G}$ after the second iteration.]{
        \centering
        \includegraphics[scale=.7]{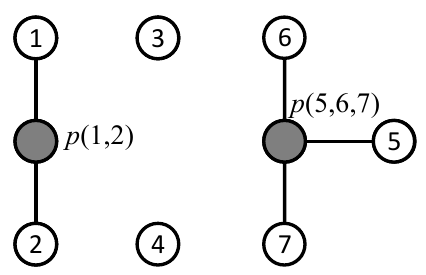}
        \label{fig: recovered tree2}
	}
    \caption{First two iterations of the tree identification algorithm. The third iteration (not shown) recovers the complete graph in Figure \ref{fig: original tree}.}
	    \label{fig: tree example}
\end{figure}

\subsection{Analysis}

In order to prove that Algorithm \ref{alg: tree} obtains the minimal tree, we first show that Step 4 identifies a node $k^*_1$ whose parent becomes a leaf node when we perform the node removal in Step 7. In Step 8 of the algorithm, this allows us to use the interference properties of the tunnels starting or ending at $k^*_1$ to obtain the interference of the tunnels of the parent node.

\begin{lemma} \label{lem: tunnel id}
Let $l = (l_1, l_2, ..., l_{|l|})$ be the tunnel that interferes with the largest number of other tunnels. When all the leaf nodes connected to $l_2$ are removed, $l_2$ becomes a leaf node in the resulting graph.
\end{lemma}
\begin{proof} The proof of the lemma is gven in Appendix \ref{app: lemma 4}.
\end{proof}

The following lemma shows that Procedure 3 identifies the nodes that share the same parent. The main idea behind the proof is that a path between two nodes that share the same parent interferes with only the tunnels starting or ending at these nodes.

\begin{lemma} \label{lem: same parent}
Two leaf nodes of a tree $i$ and $j$ share the same parent if and only if the tunnel from $i$ to $j$ does not interfere with any tunnel $l$ such that $l_1 \ne i$ or $l_{|l|} \ne j$.
\end{lemma}
\begin{proof}
The proof of this lemma is given in Appendix \ref{app: lemma 5}.
\end{proof}

\floatname{algorithm}{Procedure}
\begin{algorithm}[h]
\caption{AreSiblings$(\mathcal{F}, i,j)$ for checking whether two nodes $i$ and $j$ share the same parent}
\label{alg: isSibling}
\begin{enumerate}
\item Let $k$ be the tunnel going from node $i$ to $j$.

\item  For each tunnel $l$ in the network:

\qquad If $\mathcal{F}_{kl} = 1$ and $l_1 \ne i$ and $l_{|l|} \ne j$

\qquad \qquad Nodes $i$ and $j$ don't share the same parent.

\qquad \qquad  Return.
\item Let $k$ be the tunnel for $j$ to $i$ and repeat Step 2.
\item Nodes $i$ and $j$ share the same parent.
\end{enumerate}
\end{algorithm}
\floatname{algorithm}{Algorithm}

Now we prove the following theorem that shows that the algorithm recovers the minimal tree network. 

\begin{theorem}
If a given network $G$ is a minimal tree, then Algorithm \ref{alg: tree} recovers the network.
\end{theorem}
\begin{proof}
The proof of this theorem is given in Appendix \ref{app: theorem 4}.
\end{proof}

Note that not only the recovered graph $\hat{G}$ is isomorphic to $G$, the relative positions of the overlay nodes are the same. That is if the overlay nodes $i$ and $j$ share the same parent in $G$, they also share the same parent in $\hat{G}$. Also, because of the fact that the $\mathcal{F}$ matrix for a non minimal tree is the same as that of the minimal version of the tree, and the minimal tree is unique for any non-minimal tree, we get the following corollary.

\begin{corollary}
If a given network $G$ is a non-minimal tree, then the tree $\hat{G}$ recovered by Algorithm \ref{alg: tree} is the unique minimal tree for $G$.
\end{corollary}

The following corollary states that the graph generated by the tree algorithm solves the ILP optimally. This is true simply because all minimal trees satisfy the condition of Theorem \ref{thm: minimality}.
\begin{corollary}
If the interference pattern in a $\mathcal{F}$ matrix can be represented in a tree, Algorithm \ref{alg: tree} produces a $\hat{G}$ that solves the ILP optimally. 
\end{corollary}

Note that even when $G$ is not a tree, Algorithm \ref{alg: tree} can produce a tree as long as the interference can be represented by a tree. However if the interference pattern cannot be represented by a tree this algorithm will either fail Step 4, or the algorithm terminates but the recovered $\hat{G}$ has a different interference matrix than $\mathcal{F}$.

\section{Identifying Rings}
We now consider a situation when the $\mathcal{F}$ matrix cannot be represented in a tree. Specifically we consider a graph $G$ where the underlay nodes are arranged in a ring, and each underlay node is attached to exactly one overlay node. Also, we will assume that the network uses a shortest path routing algorithm, hence, the tunnels take the shortest paths between the overlay nodes. If the $\mathcal{F}$ matrix can be represented in a ring, our algorithm identifies the order of the overlay nodes. Note that knowing the order of the nodes gives more information than just recovering isomorphic graphs, e.g. in \cite{kelner}. Just like the tree discovery algorithm in the previous section, this algorithm can also be used to show that a particular network is not a ring.

\subsection{Algorithm}
The ring identification algorithm is given in Algorithm \ref{alg: ring}. This algorithm builds the ring in an incremental fashion. First, in Step 1 an overlay node $i$ and its parent node $p(i)$ are added to $\hat{G}$. The key idea behind the algorithm is in Step 2. It uses the $\mathcal{F}$ matrix to identify two overlay nodes in the ring that are closest to node $i$, i.e. two nodes such that their parents are neighbors of $p(i)$. In Steps 3 to 5 we attach the two nodes to their parents, and connect the parents to $p(i)$.

\begin{algorithm}[h!]
\caption{IdentifyRing$(\mathcal{F}, \mathcal{O})$ for recovering a ring network}
\label{alg: ring}
For each overlay node $i \in \mathcal{O}$:
\begin{enumerate}
\item If $i$ is not in $\hat{G}$, add two nodes $i$ and $p(i)$ to $\hat{G}$. Add an edge $\{i, p(i)\}$ to $\hat{G}$.
\item Identify two tunnels starting at node $i$ that interfere with the least number of other tunnels. Call these tunnels $k^*$ and $l^*$.
\item If node $k^*_{|k^*|}$  is not in $\hat{G}$, add two nodes $k^*_{|k^*|}$ and $p(k^*_{|k^*|})$. Add edge $\{k^*_{|k^*|}, p(k^*_{|k^*|})\}$.
\item Add edge $\{p(i), p(k^*_{|k^*|})\}$ to $\hat{G}$ if it doesn't exist.
\item Repeat Steps 3 and 4 for node $l^*_{|l^*|}$.
\end{enumerate}
\end{algorithm}

\subsection{Analysis}

We will show that Algorithm \ref{alg: ring} is guaranteed to recover the correct ring if $|\mathcal{O}| \ge 5$. For $|\mathcal{O}| = 3$, any ordering of the nodes is the same because the network links are bidirectional, so, using the algorithm is unnecessary. The algorithm might not produce the correct result for the a network with $|\mathcal{O}| = 4$ if the tunnels between the nodes in opposite sides pass through the same set of nodes. The networks in both of these situations with 3 or 4 overlay nodes in a ring are not minimal.

\begin{lemma}
Let $G$ be a graph where the underlay nodes are arranged in a ring, and each underlay node is connected to exactly one overlay node. Let $|\mathcal{O}| \ge 5$. Let $i$ and $j$ be two overlay nodes and $l$ be the tunnel from $i$ to $j$. Underlay nodes $p(i)$ and $p(j)$ are neighbors if and only if tunnel $l$ interferes with the fewest number of other tunnels.
\end{lemma}
\begin{proof}
The proof of this lemma is given in Appendix \ref{app: lemma 6}.
\end{proof}

Algorithm \ref{alg: ring} identifies overlay nodes whose parent nodes are neighbors and pieces them together into a ring. Hence, Theorem \ref{thm: ring} follows directly from the lemma above.
\begin{theorem} \label{thm: ring}
If the given network $G$ is a minimal ring, Algorithm \ref{alg: ring} recovers the network.
\end{theorem}

Similar to the tree identification algorithm, this algorithm will produce a corresponding minimal ring if the original network is a non-minimal ring. This is true because both the rings have the same $\mathcal{F}$ matrix. Also, because a minimal ring satisfies the sufficient condition for minimality, this algorithm optimally solves the ILP for ring networks. Hence we get the following two corollaries.

\begin{corollary}
If a given network $G$ is a non-minimal ring with $|\mathcal{O}| \ge 5$, then the ring $\hat{G}$ recovered by Algorithm \ref{alg: ring} is the unique minimal ring for $G$.
\end{corollary}
\begin{corollary}
If the interference pattern in the $\mathcal{F}$ matrix with $|\mathcal{O}| \ge 5$ can be represented in a ring, Algorithm \ref{alg: ring} produces a $\hat{G}$ that solves the ILP optimally. 
\end{corollary}

\section{Identifying General networks}
Inspired by the algorithms for identifying trees and rings in the previous sections, we develop a scheme for identifying general networks. A network can consist of trees and rings connected to each other. Our algorithm assumes that the network uses shortest path routing, and attempts to separate the trees from the rest of the graph, and identify these components separately. We will use Algorithm \ref{alg: tree} for recovering the trees, and we will design a new algorithm inspired by Algorithm \ref{alg: ring} for the non-tree components. Finally we will combine the discovered components to obtain the full network. This scheme is largely a heuristic, hence, we will compare its performance against another algorithm that also discovers general graphs. 
\subsection{Algorithm}
We first present Algorithm \ref{alg: multi ring} which is designed to recover a graph where every underlay node is part of one or more cycles and only one overlay node is attached to each underlay node. The algorithm works in similar fashion as the ring recovery algorithm from the previous section. The difference is that now each underlay node can have more than two underlay neighbors. So, for each overlay node $i$, the algorithm attempts to find all the overlay nodes whose parents are neighbors of $p(i)$. For clarity, we present this part of the algorithm separately in Procedure 6. 

The main idea behind Procedure 6 is shown in an example in Figure \ref{fig: three neighbors}. For Node 1, the procedure first identifies two neighbors of $p(1)$ using the tunnels that start at 1 and intersects with the fewest number of other tunnels. The intuition behind this is the same as the ring algorithm from the previous section, however, when there are more than one rings, it is not guaranteed that the shortest tunnels have the fewest number interferences. It is possible that the tunnel (1,...,5) intersects with the same number of tunnels as (1,...,3). After identifying the two neighbors, the procedure avoids any tunnels that pass through these neighbors and identifies other shortest tunnels.

\begin{figure}[h]
\centering
\includegraphics[scale=.6]{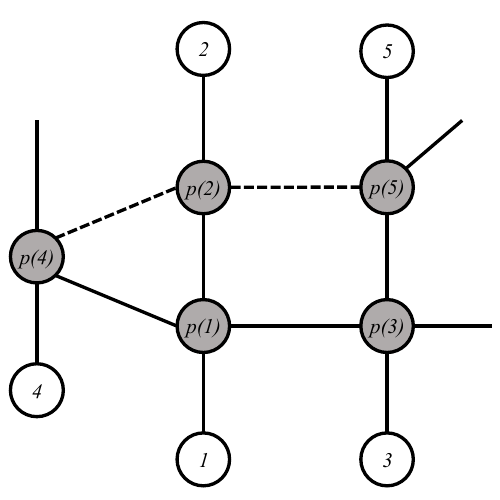}
\caption{Example of Procedure 6 at work. Node $p(1)$ has three neighbors $p(2), p(3)$ and $p(4)$. Procedure 6 first attempts to identify two nodes, e.g. 2 and 4, by minimizing the number of tunnel intersections. Then node 3 is identified by using the property that tunnel (3,...,1) doesn't interfere with the tunnels (2,...,4) or (4,..,2). }
\label{fig: three neighbors}
\end{figure}

\begin{algorithm}[h!]
\caption{IdentifyRings$(\mathcal{F}, \mathcal{O})$ for recovering a non-tree network}
\label{alg: multi ring}
Initialize $\hat{G}$ to empty graph.\\
For each overlay node $i \in \mathcal{O}$:
\begin{enumerate}
\item Obtain the to neighboring nodes, $R$ = allNeighbors($i$).
\item For each $r \in R$:
	\begin{enumerate}
		\item If node $r$ is not in $\hat{G}$ add nodes $r$ and $p(r)$  and edge $\{p(i), p(r)\}$ to $\hat{G}$.
		\item Add edge $(p(i), p(r))$ to $\hat{G}$ if it doesn't exist.
	\end{enumerate}
\end{enumerate}
\end{algorithm}

\floatname{algorithm}{Procedure}
\begin{algorithm} \label{proc: allNeighbors}
\caption{allNeighbors$(\mathcal{F}, \mathcal{O},i)$ for finding all $j$ such that $p(i)$ and $p(j)$ are neighbors}
\begin{enumerate}
\item Identify two tunnels starting at node $i$ that interfere with the least number of other tunnels. Add the end nodes of these tunnels to set $R$.
\item For each $n \in ({\mathcal{O} \backslash R})$, find a tunnel $k=(1,...,n)$ such that it interferes with the fewest number of tunnels and does not interfere with any tunnel $l$ such that $l_1, l_{|l|} \in R$.
\item If tunnel $k$ exists, add $k_{|k|}$ to $R$ and goto Step 2.
\item Return $R$.
\end{enumerate}
\end{algorithm}
\floatname{algorithm}{Algorithm}

Finally, we present Algorithm \ref{alg: general} for identifying networks with multiple rings and trees. In Step 2, this algorithm identifies sets of overlay nodes that could be a part of a tree using Procedure 3. Step 2(i) identifies the siblings, $X$, of node $i$. Step 2(ii) obtains the siblings of all the nodes in $X$. If $j$ is a sibling of $i$, then $i$ must also be a sibling of $j$. Using this property, Step 2(iii) attempts to reduce false positives. Step 2(iv) adds the nodes that are identified as part of a tree into the set of existing nodes. If some part of the tree containing the nodes in $S$ have already been identified, then these nodes must have one node in common with $S$, i.e. $S'$ exists. In such a case, nodes in $S$ is added to $S'$, otherwise $S$ is added as a new element $\mathcal{C}$. The tunnels belonging to all but one node in $S$ are removed from $\mathcal{F}$, and Step 2 is repeated on this new interference matrix. The completion of Step 2 produces the set $\mathcal{C}$ such that each element of $\mathcal{C}$ is a set of nodes that belong to the same tree.

Step 3 of the algorithm retrieves the original $\mathcal{F}$ matrix. Then in Step 4, Algorithm \ref{alg: tree} is used  on the elements of $\mathcal{C}$ to discover their corresponding trees. If the tree identification algorithm completes successfully, then all but one of the overlay nodes belonging to the tree are removed from the $\mathcal{F}$ matrix. The node that is not removed acts as an anchor node while combing the trees and the rest of the graph. In Step 5, the resulting $F$ matrix is then used in Algorithm 5 to recover the non-tree part of the graph. In order to combine a tree with the non-tree graph, in Step 6, the anchor node corresponding to the tree is found in the graph. Then in Steps 6(ii) and 6(iii), attempts are made to connect the tree to the anchor node at different locations in tree. The algorithm keeps the connection that minimizes the difference between the interference matrix of the resulting graph $\hat{G}$ and the original $\mathcal{F}$ matrix.

\begin{algorithm}[h!]
\caption{IdentifyGeneral$(\mathcal{F}, \mathcal{O})$ for recovering general networks}
\label{alg: general}
Initialize $\hat{G}$ to empty graph.
\begin{enumerate}
\item Let $\mathcal{C}$ be an empty set. Let $\mathcal{F'} = \mathcal{F}$.
\item For each $i \in \mathcal{O}$ :
	\begin{enumerate}[i.]
	\item Use Procedure 3 to find the set of nodes that share the same parent as $i$. Let $X$ be the set.
	\item For each $j \in X$ use Procedure 3 to find the set of nodes that share the same same parent as $j$. Let $X_j$ be the set.
	\item Let $S = X \cap (\cap_j X_j)$
	\item If $|S| > 1$,
		\begin{enumerate}[a)]
		\item Let $S' \in \mathcal{C}$  be a set of nodes such that $S' \cap S \ne \{\}$.
		\item If such $S'$ exists, $S' := S' \cup S$. Otherwise, $\mathcal{C} := \mathcal{C} \cup \{S\}$
		\item Let $x$ be an arbitrary node in $S$. Let $S := S \backslash \{x\}$. 	
		\item Remove tunnels $l$ from $\mathcal{F}$ if $l_1 \in S$ or $l_{|l|} \in S$. Let $\mathcal{O} := \mathcal{O} \backslash S$.
		\item Restart Step 2.
		\end{enumerate}
	\end{enumerate}
\item Let $\mathcal{F} := \mathcal{F'}$.
\item For each $S \in \mathcal{C}$:
	\begin{enumerate}[i.]
	\item Use Algorithm \ref{alg: tree} on the nodes in $S$. Let $T$ be the corresponding tree.
	\item If the algorithm fails to produce a tree, continue.
	\item Remove tunnels $l$ from $\mathcal{F}$ if $l_1 \in S$ or $l_{|l|} \in S$. Let $\mathcal{O} = \mathcal{O} \backslash S$.
	\end{enumerate}
\item Use Algorithm 5 on the remaining $\mathcal{F}$ to obtain $\hat{G}$.
\item For each tree $T \in \mathcal{T}$:
	\begin{enumerate}[i.] 
	\item Find the overlay node $i$ that is common to $T$ and $\hat{G}$.
	\item For each underlay node $j$ of $T$, add $T$ to $\hat{G}$ by replacing $i$ in $\hat{G}$ by node $j$ of the tree. Calculate the interference matrix for each $j$.
	\item For each underlay node $j$ of $T$ add $T$ to $\hat{G}$ by replacing $p(i)$ in $\hat{G}$ by node $j$ of the tree. Calculate the interference matrix for each $j$.
	\item Keep the $\hat{G}$ that produces the interference matrix closest to $\mathcal{F}$ in Steps ii and iii.
	\end{enumerate}
\end{enumerate}
\end{algorithm}

\subsection{Simulation result}
We compare the performance of Algorithm \ref{alg: general} against that of RGD1 algorithm from \cite{kelner}. For the implementation of RGD1, we obtain the exact length of each path by  using a shortest path algorithm. All links are assumed to have unit length. We choose the parameter $Rg + \tau$ to be 4. We also tried the value of 3 and 5 for this parameter, however, the performance was not as good. 

The graphs used to obtain the data for the simulation were generated to be similar to the random graphs considered in \cite{kelner}. We first generate an Erd\H{o}s-R\'enyi  random graph with parameters $\mathcal{G}(n,2/n)$. Then we find the largest connected component of the graph, and remove all the other nodes that do not belong to this component. We then attach overlay nodes to 80\% of the remaining nodes uniformly at random. Finally, we remove any underlay nodes that have degree less than 3 by using the process discussed in Section \ref{sec: minimality}. We generate 100 networks for each value of $n$, where $n = 10, 20, ..., 50$ and obtain the measurements required for both algorithms: distances for RGD1 and the $\mathcal{F}$ matrix for our algorithm. Finally, we use the measurements to recover the graphs.

The performance of the two algorithms was measured by computing the {\em edit distance} between the original graph $G$ and the recovered graph $\hat{G}$. Edit distance measures the number of links in $\hat{G}$ that needs to be added or removed in order to make it isomorphic to $G$. This metric is similar to the metric used in \cite{kelner} to obtain the asymptotic bounds of RGD1. Unfortunately, calculating the graph edit distance is an NP-hard problem, so we use an open source tool called GEDEVO \cite{gedevo} to approximate it.

\begin{figure}
    \centering
    \subfigure[Edit distances for all the iterations with n=10.]{
        \centering
        \includegraphics[scale=.36]{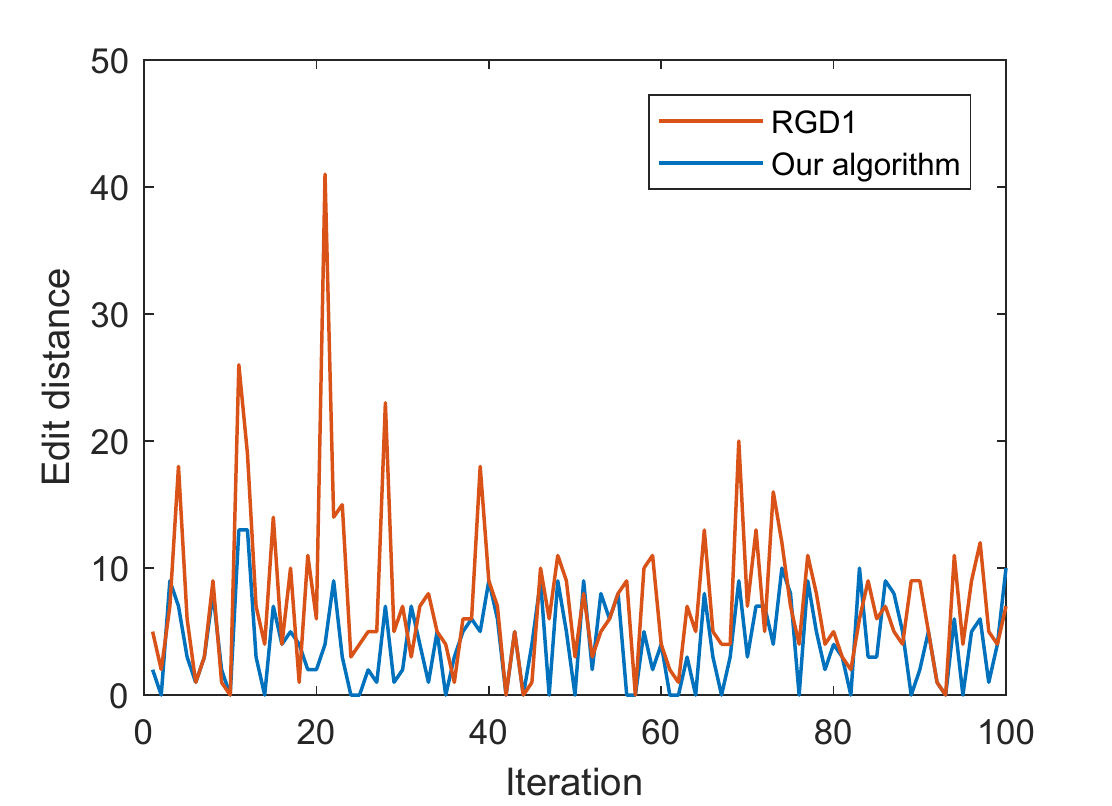}
        \label{fig: n10}
	}
    \subfigure[Edit distances for all the iterations with n=50.]{
        \centering
        \includegraphics[scale=.36]{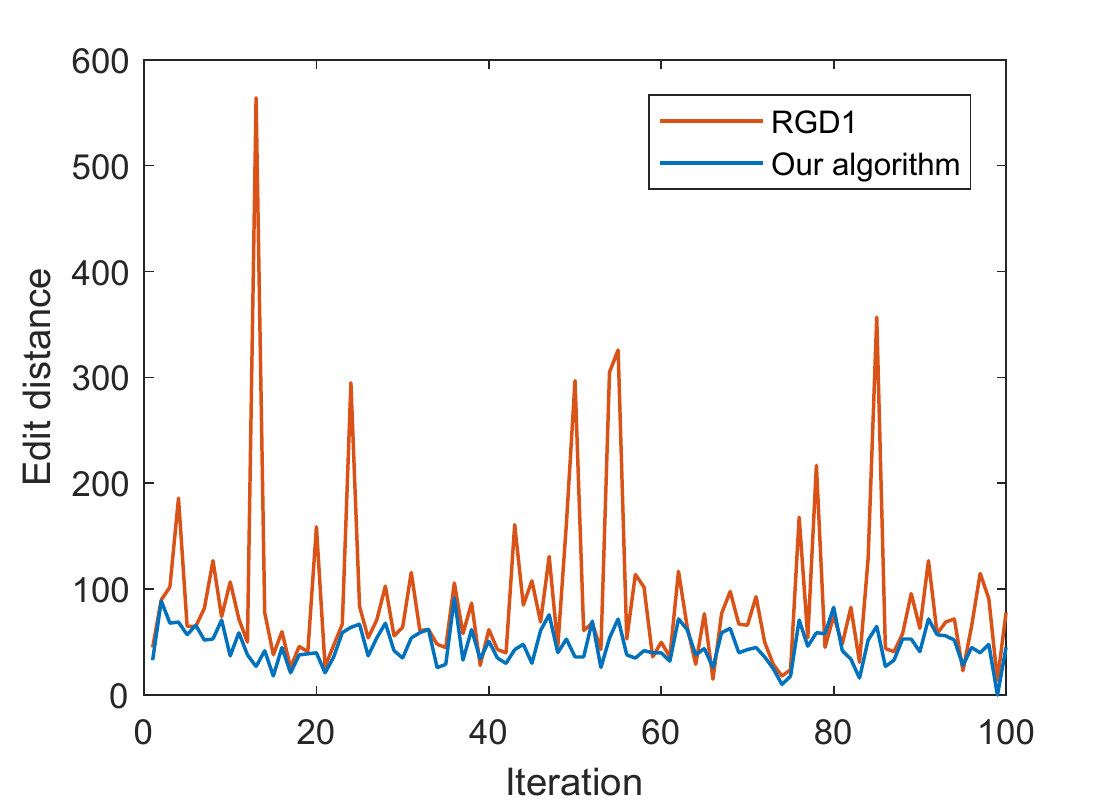}
        \label{fig: n50}
	}
 \subfigure[Average edit distance for each value of $n$.]{
        \centering
        \includegraphics[scale=.36]{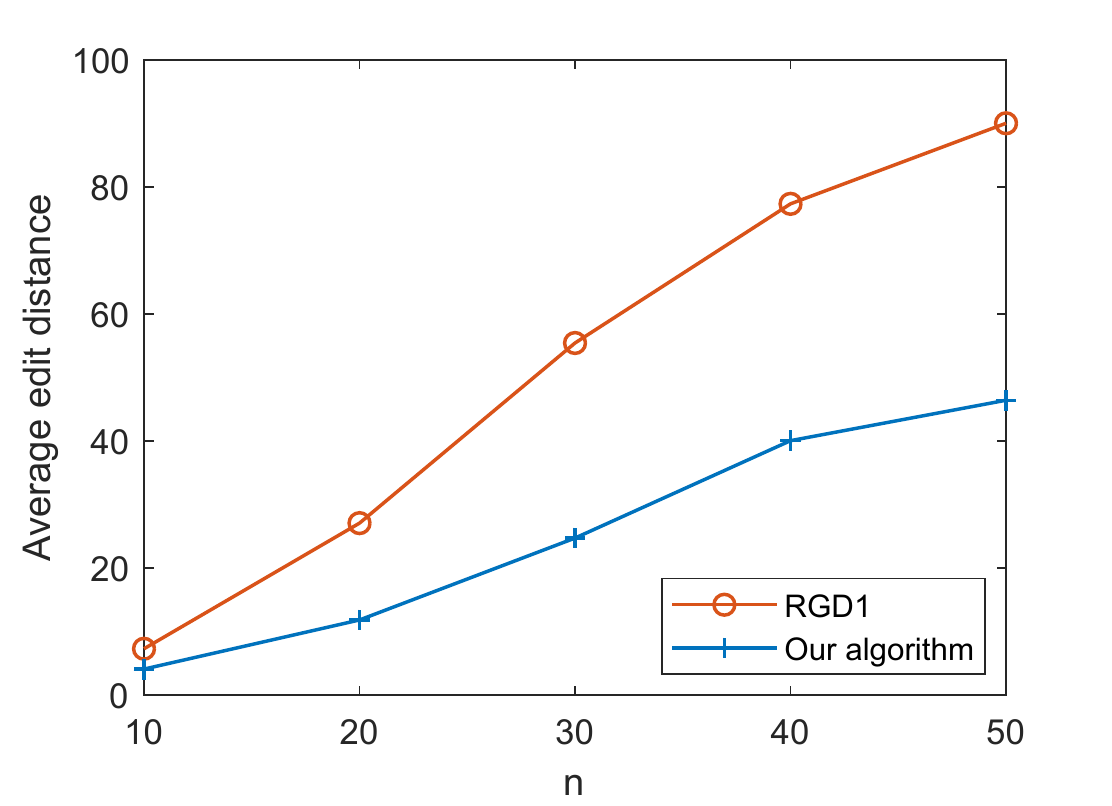}
        \label{fig: average final}
	}
    \caption{Comparison of Algorithm \ref{alg: general} and RGD1.}
	    \label{fig: simulation final}
\end{figure}

The results of the simulations are given in Figure \ref{fig: simulation final}. Figures \ref{fig: n10} and \ref{fig: n50} show the performance of the two algorithms for each of the 100 graphs that were generated. We can see that in most of the cases, our algorithm outperforms RDG1. Figure \ref{fig: average final} shows the average performance of the two algorithms across different values of $n$. Again, we can see that our algorithm outperforms RGD1.

\section{Conclusion}
We developed a new method for discovering the topology of a network. It uses the path interference information, which can be obtained by using the measurements available at the end nodes. Using the path interference, we formulated an integer linear program that finds a minimal graph that can contain all the interferences. We then developed polynomial time algorithms that solve the ILP for the special cases when the network is a tree or a cycle. Finally, we developed a heuristic for identifying general networks and compared its performance to a well known algorithm. Future research in the area will focus on developing better heuristics for general networks and providing performance guarantees.

\appendices
\section{Proof of Lemma 1}\label{app: lemma 1}
Steps 5 satisfies Constraint 2. Constraints 3 and 4 are satisfied because the initial graph formed in Step 2 is a line, and the rest of the steps never create branches or rings. Constraints 5 and 6 are satisfied because each interaction in the interference graph, i.e. the $\mathcal{F}$ matrix, is represented in one of the links in $\hat{G}$, and two tunnels that do not interfere are never assigned to the same link.

\section{Proof of Theorem 1} \label{app: theorem 1}
Lemma \ref{alg: feasible} shows that Algorithm 1 obtains a feasible solution. We need to establish the upper bound to prove the theorem.
Step 1 of the algorithm adds $|E_\mathcal{F}|$ edges to $\hat{G}$. The number of edges added by Step 3 in the worst case is upper bounded by $2L|E_{\mathcal{F}}|$ because each tunnel can require a maximum of $2|E_{\mathcal{F}}|$ extra edges. Step 5 adds $|O|$ edges, and Step 8 can add a maximum of $2L$ edges. Hence we get the required upper bound.

\section{Proof of Lemma 2} \label{app: lemma 2}
A clique $q$ in the minimum edge clique cover of a graph has at least one unique edge, i.e. an edge that is not a part of any other cliques. If this was not the case, then we can obtain a cover with fewer cliques simply by removing clique $q$. Because each edge represents an interference, each unique edge must be assigned to a different link in $\hat{G}$. 

If $|\hat{E}_D | < C$, then two unique edges of the interference graph have been assigned to the same edge of $\hat{G}$. This contains at least two tunnels that do not interfere with each other which violates the interference constraints in the ILP.

\section{Proof of Theorem 2} \label{app: theorem 2}
Given a graph with directed edges, we consider the problem of assigning the same tunnels in an undirected network. If every edge in the directed network is used by the tunnels in both direction, then $|\hat{E}_D | = 2|\hat{E}|$. That is links $(a,b), (b,a) \in \hat{E}_D$ become a single link $\{a,b\} \in \hat{E}$. However, in the directed network, some of the links can be used only in one direction. Hence, $|\hat{E}_D | \le 2|\hat{E}|$. The result follows directly from Lemma \ref{lem: lower bound}.

\section{Proof of Lemma 3} \label{app: lemma 3}
First we show that if there exists tunnel pairs $k^{ij}$ and $l^{ij}$ that intersect at link $(i,j)$ and nowhere else, then $C = 2|E|$. We know that $G$ provides a feasible solution to the ILP, hence from Theorem \ref{thm: cliques bound}, $C \le 2|E|$. Also, the interference graph $G_\mathcal{F}$ has a clique corresponding to each directed edge in G as long as some flows intersect in this link. It is sufficient to show that if the condition is satisfied then each clique in   $G_\mathcal{F}$ corresponds to a unique link $G$.

Let $Q_{ij}$ be the clique corresponding to the directed link $(i,j)$. $Q_{ij}$ has a link between the nodes $k^{ij}$ and $l^{ij}$, and this link is not part of any other clique. Hence, $Q_{ij}$ must be a clique in the minimum edge clique cover. This shows that there is one to one correspondence between the cliques in the minimum edge clique cover of $G_\mathcal{F}$ and the links of $G$.

Next we show that if $C = 2|E|$ then there exist tunnel pairs $k^{ij}$ and $l^{ij}$ that intersect at link $(i,j)$ and nowhere else. Let $L_{ij}$ be the set of all the tunnels that pass through at link $(i,j)$. Note that $L_{ij}$ must have at least two tunnels because if $L_{ij}$ has less than two tunnels then there is no clique corresponding to link $(i,j)$ giving $C < 2|E|$.

For contradiction, assume that every pair of tunnels $(k,l) \in L_{ij}$ also intersects at some other link $(x^{kl},y^{kl}) \in E$. Now we can consider a set of cliques corresponding to every link in the network other than link $(i,j)$ and cover all the edges in the interference graph giving $C< |E|$.

\section{Proof of Theorem 3} \label{app: theorem 3}
We know from Theorem \ref{thm: cliques bound} that $C \le 2|\hat{E}^*|$. We also know that the original network provides a feasible solution to the ILP, so $|\hat{E}^*|\le |E|$. Hence, $$C \le 2|\hat{E}^*|\le 2|E|.$$ 

By Lemma \ref{lem: cliques equality}, when the condition in the theorem statement is satisfied $C = 2|E|$. Hence by a sandwiching argument $|\hat{E}^*| = |E|$. 

\section{Proof of Lemma 4} \label{app: lemma 4}
The proof is by contraction. Assume that $l$ interferes with the most number of other tunnels, but when all the siblings of $l_1$ are removed $l_2$ is not a leaf node. Because of this assumption, $l_2$ has at least one neighbor node $n \not\in l$ such that $n$ is not a leaf node as shown in Figure \ref{fig: tree proof}. Since $G$ is a minimal tree the subtree of $n$, formed by removing the link $(n, l_2)$, has at least two leaf nodes $n_1$ and $n_2$. 

\begin{figure}[h]
\centering
\includegraphics[scale = 1]{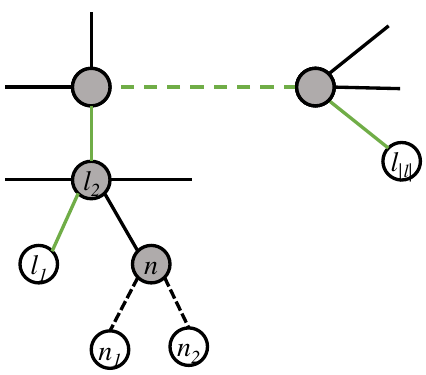}
\caption{If tunnel $l$ interferes with the most number of tunnels then all the nodes connected to $l_2$ must be leaf nodes. Otherwise, there exists a tunnel $(n_1, n, ..., l_{|l|})$ that interferes with more tunnels than tunnel $l$.}
\label{fig: tree proof}
\end{figure}

Consider a graph $G'$ formed by removing the neighbors of node $n$ other than $l_2$. In this graph, because of symmetry, a tunnel from $n$ to $l_{|l|}$ interferes with the same number of tunnels as $l$. Clearly, in graph $G$, the tunnel from $n_1$ to  $l_{|l|}$ interferes with more tunnels because in addition to all the tunnels that the path from $n$ to $l_{|l|}$ interferes with it also interferes with the tunnel from $n_1$ to $n_2$. This leads to a contradiction.

\section{Proof of Lemma 5} \label{app: lemma 5}
If $i$ and $j$ share the same parent, the path from $i$ to $j$ contains only two links $(i,p(i))$ and $(p(i), j)$. None of the tunnels that don't start in $i$ or end in $j$ use these links, hence no such tunnels $l$ intersect with the tunnel from $i$ to $j$.

If $i$ and $j$ do not share the same parent, $p(i)$ must be connected to a leaf node $i'$ in the subgraph obtained by removing the link $\{l_2, l_3\}$. Similarly $p(j)$ must be connected to a leaf node $j'$ in the subgraph obtained by removing the link $\{l_{|l|-2}, l_{|l|-3} \}$. The tunnel connecting the nodes $i'$ and $j'$ intersects with the tunnel connecting $i$ and $j$.

\section{Proof of Theorem 4} \label{app: theorem 4}
We will need one more lemma before proving the main theorem. This lemma simply uses Lemma \ref{lem: same parent} to show that Step 5 identifies the correct set of nodes.

\begin{lemma} \label{lem: group sibling}
Consider the set of nodes $X_{k^*}$ obtained in Step 5 of Algorithm \ref{alg: tree}. A leaf node $i$ is in $X_{k^*}$ if and only if $i$ shares the same parent as $k^*_1$.
\end{lemma}
\begin{proof}
By Lemma \ref{lem: same parent}, node $i$ and $k^*_1$ pass the test of Algorithm \ref{alg: isSibling} if and only if they share the same parent. Step 3 collects all the nodes that pass the test into $X_{k^*}$ and ignoring any node that doesn't. Hence, we obtain the required set  $X_{k^*}$.
\end{proof}

\subsection{Proof of Theorem 4}
By Lemmas \ref{lem: tunnel id} and \ref{lem: group sibling} we can see that Steps 5 and 6 identify a group of sibling nodes such that removing them makes their parent a leaf node. Steps 7 and 8 produce a the $\mathcal{F}$ matrix of a tree with the siblings of $k^*_1$ pruned. The new $\mathcal{F}$ matrix so formed corresponds to the such a tree because interference of tunnels starting or ending on the node $p(k^*_1)$ is exactly the same the tunnels starting or ending at node $k^*_1$ when its siblings are removed. Meanwhile, the pruned portion tree is recreated in $\hat{G}$ at every iteration of Step 5. Hence, when all the nodes in G are removed, the complete graph is created in $\hat{G}$.

\section{Proof of Lemma 6} \label{app: lemma 6}
Let $\mathcal{O} = \{1,2,...,n\}$ where $n = |\mathcal{O}|$. Assume that the correct ordering of the nodes in the ring is $p(1), p(2), ..., p(n)$. We want to show that tunnels from node 1 to 2  and 1 to n intersect with fewer tunnels than any other tunnel that start at node 1. 

We begin by showing that the tunnel $(1,...,2)$ intersects with fewer tunnels than tunnel $(1, ..., 3)$. These tunnels share the links $(1,p(1))$ and $(p(1), p(2))$. So any tunnel passing through these links intersect with both the tunnels. Also, because of symmetry, the number of tunnels intersecting with tunnel $(1,...,2)$ only at link $(p(2),2)$ is equal to the number of tunnels intersecting with tunnel $(1, ..., 3)$ only at link $(p(3),3)$. The tunnel $(2,...,4)$ does not intersect with tunnel $(1,...,2)$ however it intersects with tunnel $(1, ..., 3)$ only at link $(p(2), p(3))$. Hence tunnel  $(1, ..., 3)$ intersects with at least one more link than tunnel $(1,...,2)$. Clearly, any longer tunnel starting at node 1 must interfere with even more tunnels.
 
The ring has at least 5 nodes and the network is using the shortest path routing, so we can apply the same argument as above to show that to show that $(1,...,n)$ also intersects with the fewest number of tunnels among the tunnels starting at node 1 and passing through link $(p(1), p(n))$. Since all tunnels that start at node $1$ has to pass through either $p(n)$ or $p(1)$, these two tunnels must be the ones that intersect with the fewest other tunnels that start at node 1. 

Because of symmetry this property holds for tunnels starting at every node in the network. This completes the proof.

\begin{thebibliography}{99}
\bibitem{vardi} Vardi, Y., 1996. Network tomography: Estimating source-destination traffic intensities from link data. Journal of the American statistical association, 91(433), pp.365-377.

\bibitem{survey} Castro, R., Coates, M., Liang, G., Nowak, R., and Yu, B. Network tomography: Recent developments. Statistical science, 2004.

\bibitem{rocketfule} Spring, Neil, Ratul Mahajan, and David Wetherall. "Measuring ISP topologies with Rocketfuel." ACM SIGCOMM Computer Communication Review 32.4 (2002): 133-145.

\bibitem{crovella} B. Donnet, P. Raoult, T. Friedman and M. Crovella, "Deployment of an Algorithm for Large-Scale Topology Discovery," in IEEE Journal on Selected Areas in Communications, Dec. 2006.

\bibitem{sarac} Gunes, Mehmet Hadi, and Kamil Sarac. "Resolving anonymous routers in internet topology measurement studies." INFOCOM 2008. The 27th Conference on Computer Communications. IEEE. IEEE, 2008.


\bibitem{nowak} Rabbat, Michael, Robert Nowak, and Mark Coates. "Multiple source, multiple destination network tomography." INFOCOM 2004. Twenty-third AnnualJoint Conference of the IEEE Computer and Communications Societies. Vol. 3. IEEE, 2004.

\bibitem{nowak2} Mark Coates, Rui Castro, Robert Nowak, Manik Gadhiok, Ryan King, and Yolanda Tsang. 2002. Maximum likelihood network topology identification from edge-based unicast measurements. In Proceedings of the ACM SIGMETRICS, 2002. 

\bibitem{nowak3} Mark Coates, Michael Rabbat, and Robert Nowak. Merging logical topologies using end-to-end measurements. In Proceedings of the 3rd ACM SIGCOMM, 2003. 

\bibitem{towsley} N. Duffield, F. Lo Presti, V. Paxson and D. Towsley, "Network loss tomography using striped unicast probes," in IEEE/ACM Transactions on Networking, vol. 14, no. 4, pp. 697-710, Aug. 2006.

\bibitem{yang} J. Ni, H. Xie, S. Tatikonda and Y. R. Yang, "Efficient and Dynamic Routing Topology Inference From End-to-End Measurements," in IEEE/ACM Transactions on Networking, Feb. 2010.

\bibitem {kelner} Animashree Anandkumar, Avinatan Hassidim, and Jonathan Kelner. 2011. Topology discovery of sparse random graphs with few participants. SIGMETRICS Perform. Eval. Rev. 39, 1 (June 2011)

\bibitem {singh} A. Krishnamurthy and A. Singh, "Robust multi-source network tomography using selective probes," 2012 Proceedings IEEE INFOCOM, Orlando, FL, 2012, pp. 1629-1637.

\bibitem{sattari} P. Sattari, M. Kurant, A. Anandkumar, A. Markopoulou and M. Rabbat, ``Active learning of multiple source multiple destination topologies,'' Annual Conference on Information Sciences and Systems (CISS), 2013.

\bibitem{TSP} C. E. Miller, A. W. Tucker, and R. A. Zemlin. 1960. Integer Programming Formulation of Traveling Salesman Problems. J. ACM 7, 4 (October 1960), 326-329.

\bibitem{optimization_book} D. Bertsekas, J. Tsitsiklis, ``Introduction to linear optimization'', Athena Scientific, 1997.

\bibitem{gurobi} Gurobi Optimizer, Gurobi optimization, \url{http://www.gurobi.com/}.

\bibitem {roberts} Roberts, Fred S., ``Applications of edge coverings by cliques,'' Discrete Applied Mathematics, 1985.

\bibitem {ron} D. Andersen, H. Balakrishnan, F. Kaashoek, and R. Morris. Resilient overlay networks. In Proc. ACM SOSP, Oct. 2001.

\bibitem {jones} N. M. Jones, G. S. Paschos, B. Shrader and E. Modiano, "An Overlay Architecture for Throughput Optimal Multipath Routing," in IEEE/ACM Transactions on Networking, vol. 25, no. 5, pp. 2615-2628, Oct. 2017.

\bibitem{survivability} K. Lee, E. Modiano and H. W. Lee, "Cross-Layer Survivability in WDM-Based Networks," in IEEE/ACM Transactions on Networking, vol. 19, no. 4, pp. 1000-1013, Aug. 2011.

\bibitem{gedevo} Rashid Ibragimov, Maximilian Malek, Jiong Guo, Jan Baumbach: GEDEVO: An Evolutionary Graph Edit Distance Algorithm for Biological Network Alignment. \url{http://gedevo.mpi-inf.mpg.de/}
\end{thebibliography}
\end{document}